\documentclass[draft]{birkjour}
\usepackage[noadjust]{cite}
\usepackage{url}
\usepackage[T2A]{fontenc}
\usepackage[russian,main=english]{babel}    
\usepackage{amsmath,amssymb}
\usepackage{multirow}
\usepackage{amsfonts}
\usepackage{cite}
\usepackage{hyperref}
\usepackage{caption}
\input xy
\xyoption{all}
\newtheorem{thm}{Theorem}[section]
\newtheorem{cor}[thm]{Corollary}
\newtheorem{lem}[thm]{Lemma}
\newtheorem{prop}[thm]{Proposition}
\theoremstyle{definition}
\newtheorem{defn}[thm]{Definition}

\theoremstyle{remark}

\numberwithin{equation}{section}

\newcommand{\Cl}{C \kern -0.1em \ell}  

\DeclareMathOperator{\Tr}{\mathrm{Tr}}

\newcommand{\BZ}{\mathbb{Z}}
\newcommand{\BR}{\mathbb{R}}
\newcommand{\BC}{\mathbb{C}}

\newcommand{\ed}{\end{document}}

\newcommand{\SU}{\mathrm{SU}}

\newcommand{\SL}{\mathrm{SL}}

\newcommand{\Spin}{\mathbf{Spin}}

\newcommand{\su}{\mathfrak{su}}
\newcommand{\g}{\mathfrak{g}}
\newcommand{\spin}{\mathfrak{spin}}

\renewcommand{\sl}{\mathfrak{sl}}

\newcommand{\bra}[1]{\left \langle #1 \right |}
\newcommand{\ket}[1]{\left | #1 \right \rangle}

\begin{document}

%
%
%
%
%
%
%
%
%

\title[Isomorphisms of $\Spin\left( \frac{1}{2}\right) $ to $\SU(1,1)-\mbox{Boson}$]{ Isomorphisms of $\Spin\left( \frac{1}{2}\right) $ to $\SU(1,1)-\mbox{Boson}$:\\ Universal Enveloping and Kangni-type Transformation}%

\author[ F. A. Howard]{ Francis Atta Howard}

\address{
	University of Abomey-Calavi,\\
	International Chair in Mathematical Physics and Applications\\ 
	(ICMPA--UNESCO Chair),
	072 B.P. 50 \\ 
	Cotonou,   Benin Republic}
\email{ hfrancisatta@ymail.com; hfrancisatta@gmail.com}


\author[K. Kangni]{Kinvi Kangni}
\address{ 
	University of Felix Houphouet Boigny, \\
	UFR- Mathematiques et Informatique,\\ 
	22 BP. 1214, Abidjan 22, C\^{o}te d'Ivoire}
\email{kangnikinvi@yahoo.fr}


\keywords{Haar measure; Spherical Fourier transforms; $\SU(1,1)$-quasi boson; Universal envoloping algebra; Hopf Structure; Spin particle }

%
\begin{abstract}
In this study we investigate the nexus between the $\Spin (\frac12)$ and the $\SU(1,1)$-quasi boson Lie structure and reveal related properties as well as some decomposition of spin particles. We show that the $\SU(1,1)$-quasi boson has a left invariant Haar measure and we ascertain its spherical Fourier transformation. We finally show that this spherical Fourier transformation of type delta is a Kangni-type transform when the Planck's constant, $\hbar=1$.\\

Dans cette recherche, nous explorons le lien entre la structure de Lie des quasi-bosons $\SU(1,1)$ et $\Spin (\frac12)$, mettant en lumière certaines propriétés associées ainsi que la décomposition de particules de spin. Nous démontrons que le quasi-boson $\SU(1,1)$ possède une mesure de Haar invariante à gauche et nous déterminons sa transformation sphérique de Fourier. Nous démontrons finalement que cette transformation de Fourier sphérique type delta est une transformation de type Kangni lorsque la constante de Planck, $\hbar=1$.
\end{abstract}

\maketitle

\section{Introduction}
The $\SL(2,\BR)$ and $\SU(1,1)$ \cite{ui1970clebsch} Lie groups are two elementary groups which are very important in mathematics and have several applications in Physics. In elementary particle physics these groups arise many uniques fields;\\ Schwinger's  realization of $\su(1,1)$ Lie algebra  with creation and annihilation operators \cite{schwinger1965quantum} was defined with spatial reference  in  the Pauli matrix representation. 
Elementary spin particles have Lie structure which are parastatistics elements with some kind of Hopf alagbras. These Lie algebras have its corresponding Lie groups which are specifically Spin Lie groups, that is, Fermion Spin Lie group and Boson Spin Lie group \cite{hounkonnou2022group}. A recent study in \cite{hounkonnou2022group} by Hounkonnou, Howard and Kangni, showed that these Spin Lie groups arise from Clifford algebras and they are connected and semisimple. They further showed that any Spin Lie group, $\mathrm{G}$ can be decomposed into  
$$
\mathrm{G}= \mbox{\CYRZH} K D^{s} N
$$ 
where $K$ is compact, $D^{s}$ is a rotational function ($d$-function), and $N$ is nilpotent (Ladder operators) and  \CYRZHDSC $(\alpha^{-1})$ denote the fine structure constant and all other translational energy of elementary spin particles. This decomposition reduces to the Iwasawa decomposition when the fine structure constant  \CYRZHDSC=1, and the $d-$ function is $D^{\frac{1}{2}}$.
Several authors including Drinfeld\cite{drinfeld1989almost, kulish1981quantum, jimbo10q, woronowicz1991unbounded} have investigated into the the quantum universal envolping algebra $\sl_{q}(2,\BR)$. These algebras have unique Hopf algebraic structures which reveal more interesting properties about it. 
Motivated by all of the above mentioned work, we prove, in this paper,  that the $\Spin\left(\frac12\right)$ Lie group is isomorphic to the $\SU(1,1)$-quasi boson\cite{schwinger1965quantum} and we look at some universal enveloping algebra \cite{drinfeld1989almost, jimbo10q, kulish1981quantum, woronowicz1991unbounded} of the spin half and then consider some general application to spherical Fourier transfromations of type delta\cite{kangni2001transformation, kangni1996transformation}.

The  paper is organized as follows. In section \ref{prelim}, we recall main  definitions  and known results  useful in the sequel, and set the notation. 
Section \ref{Univ} deals with the Universal enveloping algebra $\mathcal{U}\left(\spin\left(\frac12\right)\right)$. In Section~\ref{boson}, we look at isomorphism from $\Spin (\frac12)$ to $\SU(1,1)-$quasi boson. We look at some applications by constructing the Haar measure for a quasi-boson. Finally, we end with some application of Kangni-type spherical Fourier transform of the type delta to spin particles in section \ref{SFT} and with some concluding remarks in Section~\ref{rem}.

\section{Preliminaries}\label{prelim}
The group $\SU(1,1)$ is the group of two-dimensional unitary unimodular matrices which leave the form $|x_{1}|^{2}-|x_{2}|^{2}$ invariant \cite{louck1981angular, biedenharn1965unitary}. Now for a fixed choice of the phase, a $2\times2$ matrix representation ($d$-function) of $\exp(-it K_{y})$ will be: 
\begin{align}\label{rota}
	\exp(-it K_{y})
	&=\exp\left[-i\dfrac{t}{2}
	\begin{pmatrix}
		0 & i\\
		{i} & {0}
	\end{pmatrix}\right]=\exp\left[\dfrac{t}{2}
	\begin{pmatrix} 
		0 & 1\\
		{1} & {0}
	\end{pmatrix}\right]\\
	&=\exp\left(\dfrac{t}{2}\sigma_{x}\right)
	=\begin{pmatrix}
		\cosh\frac{t}{2} & \sinh\frac{t}{2}\\[0.5ex]
		\sinh\frac{t}{2} & \cosh\frac{t}{2}
	\end{pmatrix}=d^{\frac12}_{t}.
\end{align}
Let the Lie group $\mathrm{G}$ operate on the multiplicative group $U$ of the complex numbers with modulo 1:
\begin{gather}
	\zeta\longrightarrow g\cdot \zeta, \quad or \quad g\cdot \zeta =\frac{\alpha\zeta+\beta}{\bar{\beta}\zeta+\bar{\alpha}}, \quad if \quad g=\begin{pmatrix}
		\frac{\alpha}{\beta}&\frac{\beta}{\alpha}
	\end{pmatrix};
\end{gather}
if we define the Haar measure of $U$, by the formula
\begin{gather*}
	d\zeta=\frac{1}{2\pi}d\theta,\quad for \quad \zeta=\exp(i\theta), \quad 0\leq\theta<2\pi,
\end{gather*}
and if we put
\begin{align*}
	\exp(t(g,\zeta))=|\bar{\beta}\zeta+\bar{\alpha}|^2, \quad for \quad \zeta\in U, \quad and \quad g=\begin{pmatrix}
		\frac{\alpha}{\beta}&\frac{\beta}{\alpha}
	\end{pmatrix}\in G,
\end{align*}
we have, the following formula:
\begin{align}\label{equ3.4}
	d(g.\zeta)=\exp(-t(g,\zeta))d\zeta.
\end{align}
Also put, for $g\in G $ and $\zeta\in U$,
\begin{gather*}
	u(g,\zeta)=\dfrac{(\bar{\beta} \zeta+\bar{\alpha})}{|\bar{\beta} \zeta+\bar{\alpha}|}, \quad if \quad g=\begin{pmatrix}
		\frac{\alpha}{\beta}&\frac{\beta}{\alpha}
	\end{pmatrix}\in G;
\end{gather*}
then we get, for $g, g'\in G$,
\begin{gather}\label{equ3.6}
	u(g g^{\prime},\zeta)=u(g,g^{\prime}\cdot\zeta)u(g^{\prime},\zeta).
\end{gather}
Let $\mathcal{H}=L^2(U, d\zeta)$ be the Hilbert space of functions $\varphi(\zeta)$ defined on $U$, of square integrable measure $d\zeta$, with scalar product
\begin{gather*}\label{equ3.7}
	(\varphi,\psi)=\int_U\varphi(\zeta)\overline{\psi(\zeta)}d\zeta.
\end{gather*}
Let $j=0, \frac{1}{2}$, and $s\in\mathbb{C}$; we define, for $g\in G$, the operator $V_g^{j,s}$ by the following formula \cite{takahashi1961fonctions}:
\begin{gather}\label{equ2.10}
	(V_g^{j,s}\varphi)(\zeta)=\exp(-st(g^{-1},\zeta))(u(g^{-1},\zeta))^{2j}\varphi(g^{-1}\cdot\zeta);
	\\ \nonumber
	=\mid \bar{\beta}\zeta + \bar{\alpha}\mid^{-2s} \left( \dfrac{\bar{\beta} \zeta+\bar{\alpha}}{|\bar{\beta} \zeta+\bar{\alpha}|}\right)^{2j}\varphi\left( \dfrac{\alpha \zeta+\beta}{\bar{\beta} \zeta+\bar{\alpha}}\right)
\end{gather}
for $\varphi\in\mathcal{H}$ it is clear, with the help of formula \eqref{equ3.6}, to verify that we have 
\begin{equation}
	V_{g g^{\prime}}^{j,s}=V_g^{j,s}V_{g^{\prime}}^{j,s}, \quad for \quad g,g^{\prime} \in  \mathrm{G};
\end{equation}
moreover, the formula \eqref{equ3.4} shows that the operator $V_g^{j,s}$ is a unitary if $\Re(s)=\frac{1}{2}$. For any integer $p$, put
\begin{align}\label{equ3.10}
	\varphi_p(\zeta)=\zeta^{-p}.	
\end{align}
It is obvious that the functions $\varphi_p$, $p\in\BZ$, form an orthogonal base of $\mathcal{H}$, and that we obtain
\begin{gather}\label{equ2.13}
	V^{j,s}_{u_\theta}\varphi_p=\chi_{p+j}(u_\theta)\varphi_p,
\end{gather}
that is, the function $\varphi_p(\zeta)$ has weight $p+j$ (with respect to the Cartan subgroup $K$).

We denote by $\mathcal{H}_0$ the subspace of $\mathcal{H}$ formed by the functions $\varphi$ such that the mapping $g\longrightarrow V^{j,s}_g\varphi$ is analytic. It is obvious that all $\varphi_p \in \mathcal{H}_0$. The formula
\begin{gather}\label{equ2.14}
	V^{j,s}_\alpha\varphi=\int_{\mathrm{G}}V^{j,s}_{g}\varphi  d\alpha(g), \quad for \quad \varphi\in\mathcal{H}_0, \quad \alpha\in\mathcal{U},
\end{gather}
defines a representation of the univeral envelopping algebra $\mathcal{U}$ in $\mathcal{H}_0$, particularly, if $S\in \g$, we have
\begin{align}\label{equ3.12}
	V^{j,s}_{S}\varphi(p)=\underset{t\longrightarrow 0}{\lim}\frac{1}{t}(V^{j,s}_{\exp(tS)}\varphi-\varphi) \quad for \quad \varphi\in\mathcal{H}_0.
\end{align}


\subsection{Fourier Transform of Type Delta}
Let $j\in\frac{\mathbb{Z}}{2}$ be a quantum state of a fermion particle, and let the function $\chi_{j}(u_\theta)=e^{ij\theta}$
be the character of the compact abelian group $K=\left\lbrace u_{\theta}|0\leq\theta<4\pi\right\rbrace $. Let $L(G)$ be the subspace of complex-valued continuous functions with compact support satisfying:
\begin{equation*}
	f(u_\theta gu_\phi)=\chi_j(u_\theta) f(g)\chi_j(u_\phi)
\end{equation*}
for an element $g$ in $\SU(1,1)$-quasi boson and  $u_\theta, u_\phi\in K$ be denoted by $A_j$. We have for $f\in A_{j}$ satisfying:
$$f(u_\theta d^{\frac{1}{2}}_{t}u_\phi)=\chi_{j}(u_{\theta+\phi})f(d^{\frac{1}{2}}_{t}), \quad \mbox{if} \quad g=u_\theta d^{\frac{1}{2}}_{t}u_\phi,$$
such that $d^{-\frac{1}{2}}_{t}=u_\pi d^{\frac{1}{2}}_{t}u_{-\pi}$, we have $f(d^{-\frac{1}{2}}_{t})=f(d^{\frac{1}{2}}_{t})$ for $t\in \BR$. We can then 
consider $f(d^{\frac{1}{2}}_{t})$ as a function of $\cosh t$:
\begin{equation*}
	f[x]=f(d^{\frac{1}{2}}_{t})=f[\cosh t]
\end{equation*}
is a continuous function with compact support defined in $1
\leq x < +\infty$.
For $s\in\mathbb{C}$ one can put
$$	\alpha_{n,s}(g)=\overline{\chi_{j}(u_{\theta})}\exp({-st}) 
\quad \mbox{if} \quad g=u_\theta a_t\chi_\xi$$ is the unique decomposition of $g$;
\begin{gather}
	\zeta_{n,s}(g)=(\alpha_{n,s})^0(g)=\int_k\alpha_{n,s}(k^{-1}g k)dk.
\end{gather}
Let $K$ be a compact subgroup of $\mathrm{G}$, and denote by $\hat{K}$ the collection of all equivalence classes of irreducible unitary representations of $K$.
For every class $\delta$ of $\hat{K}$, we denote $\xi_\delta$ as the character of $\delta$, $d(\delta)$ as the degree of $\delta$, and define $\chi_{\delta} = d(\delta)\xi_{\delta}$. If $\delta$ represents the class of contragredient representations of $\delta \in \hat{K}$, then $\chi_{\delta} = \chi_{\hat{\delta}}$. Utilizing the Schur orthogonality relations, we can verify that


$\chi_{\hat{\delta}} * \chi_{\hat{\delta}} = \chi_{\hat{\delta}}$. For all function $f \in K(\mathrm{G})$, the algebra of continuous functions with compact support, we set
\begin{equation*}
	_\delta f(x)=\hat{\chi_{\delta}}*f(x)= \int_K \chi_\delta(k)f(kx)dk
\end{equation*}

\begin{equation*}
	f_{\delta}(x)=f(x)*\chi_{\delta}=\int_{K}\chi_{\delta}(k^{-1})f(xk)dk
\end{equation*}

where $dk$ is a normalized Haar measure on $K$.  $$K_{\delta}(\mathrm{G})=\{f\in K(\mathrm{G}),f=_\delta f=f_{\hat{\delta}} \}$$ and 
$K_{\delta(G)}$ is the subalgebra of $K(G)$, and the mapping $\chi_{\hat{\delta}} *f* \chi_{\hat{\delta}}$ is a projection of $K(G)$ onto $K_{\delta}(G)$.
Consider a Banach representation $U$ of $G$ on a Banach space $E$ \cite{n2020some ,warner2012harmonic}. Put $P(\delta)=U(\tilde{\chi_\delta})$ and $E(\delta)=P(\delta)E$, $E(\delta)$ the closed subspace of $E$ consisting of those vectors in $E$ which transform under $K$ according to $\delta$. 

\subsection{Kangni-Type Transform [KTT]}
Let $E$ be a finite dimensional complex vector space. A spherical function $\phi$ of type $\delta$ is a quasi-bounded continuous function on $G$ with values in $End_{\BC}(E)$ such that:
\begin{enumerate}
	\item[i]$\phi(kxk^{-1})=\phi(x)$
	\item[ii]$\chi_{\delta}* \phi=\phi=\phi* \chi_{\delta}$
	\item[iii]The mapping $u_{\phi}:f\rightarrow \phi(f)=\int_{G}f(x)\phi(x^{-1})dx$
\end{enumerate}
is an irreducible representation algebra of $\mathcal{K}_{\delta}^{\sharp}(G)$ \cite{kangni1996transformation, kangni2001transformation, n2020some}. The dimension of $E$ is the height of $\phi$.
If $\phi$ is a quasi-bounded continnous function on $G$ with values in $End_{\BC}(E)$ such that $\phi_{K}=\phi \quad and \quad \chi_{\delta}*\phi=\phi$.
Then the function $\phi$ is spherical function of type $\delta$ if and only if  $$\int_{K}\phi(kxk^{-1}y)dk=\phi(x)\phi(y), \quad for \quad all \quad x,y\in G.$$
Let $\delta\in \bar{K}$ and $\mu_{\delta}\in\delta$ be a unitary irreducible representation of $K$ onto the Hilbert space $E_{\delta}$. For every $f\in K_{\delta}^{\sharp}(G)$. Consider the integral defined by
\begin{gather*}
	F_{f}^{\delta}(h)=h^{\rho}\int_{K}\int_{N}f(khn)\mu_{\delta}(k^{-1})d_{N}(n)dk, \thickspace h\in A.
\end{gather*} 
We shall call the map $f\rightarrow F_{f}^{\delta}$ the Abel transformation of type $\delta$ on $G$. $K_{\delta}^{\sharp}$ is isomorphic to $U_{c,\delta}(G)$ under the map $f\rightarrow \psi_{f}^{\delta}$ defined by $\psi_{f}^{\delta}(x)=\int_{K}\mu_{\delta}(k^{-1})f(kx)dk$. Then, for every $f\in K_{\delta}^{\sharp}$, we have
\begin{gather*}
	F_{f}^{\delta}(h)=h^{\rho}\int_{K}\psi_{f}^{\delta}(hn)d_{N}(n), h\in A.
\end{gather*}
The Abel transformation is linear and one-to-one mapping of the algebra $f\in K_{\delta}^{\sharp}(G)$ onto $f\in K_{\delta}^{\sharp}(A)$. 
Let $G$ be a locally compact unimodular countable at infinity. Let $K$ be a large compact subgroup of $G$. The complexification of the Lie algebras of $G$ and $K$ are 
\begin{align*}
	g=g_{0}+ig_{0}\\
	k=k_{0}+ik_{0}
\end{align*}
Let $\mathcal{A}$ be a universal enveloping algebra of $g_{0}$, and $\mathcal{C}$ the centraliszer of $k_{0}$ in $\mathcal{A}$.

\begin{thm}\cite{kangni2001transformation}
	Let $E$ be a vector space with finite dimension on $\BC$, $\phi$ a quasi-bounded function and $K$ central class of $C^{\infty}$ function. Assume there exist an irreducible representation $u_{\phi}$ of $C$ in $E$ such that:
	$$D\phi=\phi u_{\phi}(D) \quad \mbox{where}\quad u_{\phi}(D)=D\phi(1)$$
	for all $D\in \mathcal{C}$.
	Thus, there exist $\delta\in \hat{K}$ such that $\phi$ is spherical of type delta.
\end{thm}

\begin{thm}\cite{kangni1996transformation}
	Let $\mu$ be a linear form on $\mathcal{A}$. The mapping $f\longmapsto \phi_{\delta}^{\mu}(f)$ of $K_{\delta}^{\sharp}(G)$ with value in $M_{d(\delta)}(\BC)$ defined by:
	$$\phi_{\delta}^{\mu}=\int_{A}F^{\delta}_{f}(h)h^{\mu+\rho}dh$$
	is a spherical Fourier transformation of the type delta.
\end{thm}
from the above results we have the following:
\begin{defn}\label{KTT}
	Any spherical Fourier transformation of the type delta, $\phi_{\delta}^{\mu}$ such that the mapping $f\longmapsto \phi_{\delta}^{\mu}(f)$ of $K_{\delta}^{\sharp}(G)$ with value in $M_{d(\delta)}(\BC)$ defined by:
	$$\phi_{\delta}^{\mu}=\int_{A}F^{\delta}_{f}(h)h^{\mu+\rho}dh$$
	where $\mu$ is a linear form on $\mathcal{A}$ is a Kangni-type transform.	
\end{defn}

\section{Universal Enveloping Algebra $\mathcal{U}\left(\spin\left(\frac12\right)\right)$}\label{Univ}
In this subsection we shall define the quantum $\mathcal{U}(\spin\left(\frac12\right))$ by generators and relations and then define $\Delta_{\spin\left(\frac12\right)},\epsilon_{\spin\left(\frac12\right)}, \mathbb{S}_{\spin\left(\frac12\right)}$ on these generators.

\begin{thm}\label{3.1}
	The fermionic spin Lie algebra $\spin\left(\frac12\right)$ (bosonic spin Lie algebra ($\spin 1$))
	endowed with
	\begin{enumerate}
		\item [(i)] a coprodruct $\Delta_{\spin\left(\frac12\right)}(\Delta_{\spin(1)})$, that is, a homomorphism
		\begin{align*}
			\Delta_{\spin\left(\frac12\right)} : \mathcal{U}\left(\spin\left(\frac12\right)\right)\longmapsto \mathcal{U}\left(\spin\left(\frac12\right)\right) \otimes \mathcal{U}\left(\spin\left(\frac12\right)\right)
		\end{align*}
		\begin{align*}
			\Delta_{\spin(1)} : \mathcal{U}\left(\spin(1)\right)\longmapsto \mathcal{U}\left(\spin(1)\right) \otimes \mathcal{U}\left(\spin(1)\right)
		\end{align*}
		
		\item [(ii)] a counit $\epsilon_{\spin\left(\frac12\right)}(\epsilon_{\spin(1)})$, that is, a homomorphism
		$$\epsilon_{\spin\left(\frac12\right)} : \mathcal{U}\left(\spin\left(\frac12\right)\right)\longmapsto K_{q} $$
		
		$$\epsilon_{\spin(1)} : \mathcal{U}\left(\spin(1)\right)\longmapsto K_{q} $$
		
		\item [(iii)] an antipode $\mathbb{S}_{\spin\left(\frac12\right)}(\mathbb{S}_{\spin(1)})$, that is, antihomomorphism(graded)
		$$ \mathbb{S}_{\spin\left(\frac12\right)} : \mathcal{U}\left(\spin\left(\frac12\right)\right)\longmapsto \mathcal{U}\left(\spin\left(\frac12\right)\right)$$
		
		$$ \mathbb{S}_{\spin(1)} : \mathcal{U}\left(\spin(1)\right)\longmapsto \mathcal{U}\left(\spin(1)\right)$$
	\end{enumerate}
	defined on the generators of the $\spin\left(\frac12\right) (\spin (1))$ is given by the relations:
	\begin{align*}
		\Delta_{\spin\left(\frac12\right)}(S_{+})&=S_{+}\otimes1+K\otimes S_{+}, \quad
		\Delta_{\spin\left(\frac12\right)}(S_{-})=S_{-}\otimes K^{-1}+1\otimes S_{-},\\ \nonumber
		\Delta_{\spin\left(\frac12\right)}(K)&=K\otimes K,\\
		\nonumber
		\epsilon_{\spin\left(\frac12\right)}(S_{+})&=\epsilon_{\spin\left(\frac12\right)}(S_{-})=0,\quad \epsilon_{\spin\left(\frac12\right)}(K)=1,\\
		\nonumber
		\mathbb{S}_{\spin\left(\frac12\right)}(S_{+})&=-K^{-1}S_{+}, \quad	\mathbb{S}_{\spin\left(\frac12\right)}(S_{-})=-S_{-}K, \quad	\mathbb{S}_{\spin\left(\frac12\right)}(K)=K^{-1}.
	\end{align*}
\end{thm}

\begin{proof}
	The fermionic spin Lie algebra \cite{hounkonnou2022group}, $\spin(\frac{1}{2})$ can be generated by
	\begin{enumerate}
		\item the elements 
		$${S_{-}}=\hbar
		\begin{pmatrix}
			0 & 0\\
			1 & 0
		\end{pmatrix},\quad 
		S_{z}=\frac{\hbar}{2}
		\begin{pmatrix}
			1 & 0\\
			0 & -1
		\end{pmatrix},\quad 
		S_{+}=\hbar 
		\begin{pmatrix}
			0 & 1\\
			0 & 0
		\end{pmatrix}
		$$
		
		\item the commutation relations are given by:
		\begin{align*}
			[{S_{z}},S_{+}]&=2 \hbar S_{+},\quad 	
			[{S_{z}},S_{-}]=-2\hbar S_{-},\quad \\
			\nonumber
			[{S_{+}},S_{-}]&=\dfrac{\sinh \left(\dfrac{\hbar \sigma_{z}}{2}\right)}{\sinh \left(\dfrac{\hbar}{2}\right)}=\dfrac{\sinh(S_{z})}{\sinh \left(\dfrac{\hbar}{2}\right)}=\dfrac{\exp S_{z}-\exp (-S_{z})}{\exp (\frac{\hbar}{2})-\exp(-\frac{\hbar}{2})}=K.
		\end{align*}
	\end{enumerate} 
	From the subgroups of the Iwasawa decomposition of $\spin(\frac{1}{2})$, we have that the diagonal element 
	\begin{align}
		S_{z}=\frac{\hbar}{2}
		\begin{pmatrix}
			1 & 0\\
			0 & -1
		\end{pmatrix}=\frac{\hbar}{2}\sigma_{z}.
	\end{align}
	This element can be exponentiated by mimicing the Iwasawa decomposition of a particle with the Planck's constant ($\hbar$) as a constant variable to have the following:
	\begin{align*}
		d^{\frac12}_{\hbar}&=
		\exp\left(\dfrac{\hbar\sigma_{z}}{2}\right)
		=\exp\left( \dfrac{\hbar}{2} \bra{s,m}\sigma_{z} \ket{s,m}\right)
		=\sum_{n=0}^{\infty}\dfrac{1}{n!}\left(\dfrac{\hbar\sigma_{z}}{2}\right)^{n}\notag\\
		&=\sum_{n=0}^{\infty}\dfrac{1}{n!}\left(\dfrac{\hbar\sigma_{z}}{2}\right)^{n}
		=
		\begin{pmatrix}
			\exp{(\frac{\hbar}{2})} & 0\\
			0 & \exp{(-\frac{\hbar}{2})}
		\end{pmatrix}\\
		&=\begin{pmatrix}
			q & 0\\
			0 & q^{-1}
		\end{pmatrix}=K.
	\end{align*}
	We define the algebra $\mathcal{U}$($\spin\left(\frac12\right)$) generated by $S_{+},S_{-} , K, K^{-1}$ subject to the following relations \cite{woronowicz1991unbounded, kulish1981quantum, jimbo10q}:
	
	\begin{align}
		KK^{-1}&=K^{-1}K=1, \\  KS_{+}K^{-1}&=q^2S_{+}, \\
		KS_{-}K^{-1}&=q^{-2}S_{-}, \\  S_{+}S_{-}-S_{-}S_{+}&=\frac{K-K^{-1}}{q-q^{-1}}.
	\end{align}
	Note that the algebra $\mathcal{U}(\spin\left(\frac12\right))$ is spanned by the monomials $S_{-}^{r}K^{l}S_{+}^{m}$,\\ where $r, m \in \mathbb{Z}\geqslant 0$, and $l\in \mathbb{Z}$. The quantum spin Lie algebra $\mathcal{U}\left(\spin\left(\frac12\right)\right)$ endowed with
	
	\begin{enumerate}
		\item [(i)] a coprodruct $\Delta_{\spin\left(\frac12\right)}$, that is, a homomorphism
		\begin{align*}
			\Delta_{\spin\left(\frac12\right)} : \mathcal{U}\left(\spin\left(\frac12\right)\right)\longmapsto \mathcal{U}\left(\spin\left(\frac12\right)\right) \otimes \mathcal{U}\left(\spin\left(\frac12\right)\right)
		\end{align*}
		\item [(ii)] a counit $\epsilon_{\spin\left(\frac12\right)}$, that is, a homomorphism
		$$\epsilon_{\spin\left(\frac12\right)} : \mathcal{U}\left(\spin\left(\frac12\right)\right)\longmapsto K_{q} $$
		
		\item [(iii)] an antipode $\mathbb{S}_{\spin\left(\frac12\right)}$, that is, antihomomorphism(graded)
		$$ \mathbb{S}_{\spin\left(\frac12\right)} : \mathcal{U}\left(\spin\left(\frac12\right)\right)\longmapsto \mathcal{U}\left(\spin\left(\frac12\right)\right)$$
	\end{enumerate}
	defined on the generators of the quantum algebra $\mathcal{U}\left(\spin\left(\frac12\right)\right)$ by the relations:
	
	\begin{align}
		\Delta_{\spin\left(\frac12\right)}(S_{+})&=S_{+}\otimes1+K\otimes S_{+}, \quad
		\Delta_{\spin\left(\frac12\right)}(S_{-})=S_{-}\otimes K^{-1}+1\otimes S_{-},\\ \nonumber
		\Delta_{\spin\left(\frac12\right)}(K)&=K\otimes K,\\
		\nonumber
		\epsilon_{\spin\left(\frac12\right)}(S_{+})&=\epsilon_{\spin\left(\frac12\right)}(S_{-})=0,\quad \epsilon_{\spin\left(\frac12\right)}(K)=1,\\
		\nonumber
		\mathbb{S}_{\spin\left(\frac12\right)}(S_{+})&=-K^{-1}S_{+}, \quad	\mathbb{S}_{\spin\left(\frac12\right)}(S_{-})=-S_{-}K, \quad	\mathbb{S}_{\spin\left(\frac12\right)}(K)=K^{-1}.
	\end{align}
\end{proof}

\begin{cor}\label{cor 3.2}
	The universal algebra $\mathcal{U}\left(\spin\left(\frac12\right)\right)$ of a spin particle is isomorphic to the quantum $\mathcal{U}_{q}\sl(2,\BR)$.
\end{cor}
\begin{proof}\cite{kulish1981quantum, jimbo10q, woronowicz1991unbounded}
	Proof of this follows easily from Theorem \eqref{3.1}
\end{proof}
\begin{lem}
	The Casimir operator $\omega$ of $\mathcal{U}\left(\spin\left(\frac12\right)\right)$  is given by 
	\begin{align*}
		\omega&=\dfrac{q^{-1}K+qK^{-1}}{(q-q^{-1})2}+S_{+}S_{-},
		\\
		\nonumber
		&=\dfrac{qK+q^{-1}K^{-1}}{(q-q^{-1})2}+S_{-}S_{+}.
	\end{align*}
$\omega$ commutes with the element of the $\spin(\frac12)$ Lie algebra \cite{kulish1981quantum, jimbo10q, woronowicz1991unbounded}.
\end{lem}
\begin{proof}
The proof of this lemma follows easily from corollary \eqref{cor 3.2}.
\end{proof}

\section{Isomorphism from $\Spin (\frac12)$ to $\SU(1,1)$-quasi boson}\label{boson}

\begin{defn}
	The Spin Lie group of all spin one-half particles with quantum state spanned by 2 states, $2\times2$ real matrices and determinant $1$ when $\hbar=1$ is denoted by $\Spin_{\BR}(\frac12)$ \cite{hounkonnou2022group}.
\end{defn}
The real Lie algebra $\g$ of the $\Spin_{\BR}(\frac12)$ is given by:
\begin{align}
	\spin_{\BR}\left(\frac{1}{2}\right)={\{S\in M_{2}(\BR) \mid \Tr S=0 }\}.
\end{align}

\begin{thm}\label{thm2.2}\cite{takahashi1961fonctions}
	\begin{enumerate}
		\item [(1)]Let C be a 2$\times$2 unitary matrix given by 
		\begin{align*}
			C=\frac{1}{\sqrt{2}}\left(\begin{array}{cc}
				1 & -i\\
				1 & i
			\end{array}\right)\in U\left(2\right),
			\quad and \quad let \quad g=\left(\begin{array}{cc}
				a & b\\
				c & d
			\end{array}\right)\in GL(2,\BR).
		\end{align*}
		Put $\mathrm{G}=C\cdotp g\cdotp C^{-1}$. Then $\mathrm{G}= \left\lbrace \left(\begin{array}{cc}
			\alpha & \beta\\
			\bar{\beta} & \bar{\alpha}
		\end{array}\right)\mid |\alpha|^{2}-|\beta|^{2}=1 \right\rbrace$.
		\\	If $\mathrm{G}=C\cdotp \left(\begin{array}{cc}
			a & b\\
			c & d
		\end{array}\right)\cdotp C^{-1}= \left(\begin{array}{cc}
			\alpha & \beta\\
			\bar{\beta} & \bar{\alpha}
		\end{array}\right)\in \SU(1,1)-\mbox{quasi boson} $, then 
		
		\[
		\alpha=\frac{1}{2}\left\{ \left(a+d\right)+i\left(b-c\right)\right\}, 
		\]
		\[
		\beta=\frac{1}{2}\left\{ \left(a-d\right)-i\left(b+c\right)\right\}. 
		\]
		\item [(2)]\cite{hounkonnou2022group}
		Any element $g\in\SU(1,1)-\mbox{quasi boson}$ can be uniquely decomposed into the form $$
		g= \mbox{\CYRZH} k_{\theta}d^{\frac12}_{t}n_{\xi}
		$$
		with 
		$\exp({i\frac{\theta}{2}})=\frac{\alpha+\beta}{\mid\alpha+\beta\mid}$, $\exp({t})=\mid\alpha+\beta\mid^{2}$ and $\xi=Im\left(\frac{\alpha-\beta}{\alpha+\beta}\right)$ when $\mbox{\CYRZH} =1$.\newline
		By isomorphism $w:g\longrightarrow C \cdotp g \cdotp C^{-1}$ from $\Spin_{\BR}\left(\frac{1}{2}\right)$ to $\mathrm{G}$, we obtain
		
		\begin{gather*}
			k_{\theta}=w( k_{\theta})=\left(\begin{array}{cc}
				\exp({i\frac{\theta}{2}}) & 0\\
				0 &  \exp({-i\frac{\theta}{2}})
			\end{array}\right),
			\\
			d^{\frac{1}{2}}_{t}=w\left( d^{\frac{1}{2}}_{t}\right)= \left(\begin{array}{cc}
				\cosh\frac{t}{2} & \sinh\frac{t}{2}\\
				\sinh\frac{t}{2} & \cosh\frac{t}{2}
			\end{array}\right),
			\label{eq117}
			\\
			n_{\xi}=w\left( n_{\xi}\right)= \left(\begin{array}{cc}
				1+i\frac{\xi}{2} & -i\frac{\xi}{2}\\
				i\frac{\xi}{2} & 1-i\frac{\xi}{2}
			\end{array}\right).
			\label{eq118}
		\end{gather*}
		this yields the usual Iwasawa decomposition when $\mbox{\CYRZH}=1$, that is, 
		$
		g= k_{\theta}d^{\frac12}_{t}n_{\xi}.
		$
		
		\item [(3)] By the isomorphism $w:\Spin_{\BR}\left(\frac{1}{2}\right)\longrightarrow \SU(1,1)$, we write $w( k_{\theta})=k_{\theta}$, $w\left( d^{\frac{1}{2}}_{t}\right)=d^{\frac{1}{2}}_{t}$, $w\left( n_{\xi}\right)=n_{\xi}$ and $w(gk)_{\theta}=k_{g.\theta}d^{\frac12}_{t\left(g,\theta\right)}n_{\xi\left(g,\theta\right)}$.\\
		If $g=\left(\begin{array}{cc}
			\alpha & \beta\\
			\bar{\beta} & \bar{\alpha}
		\end{array}\right)\in \SU(1,1)$, then we have
		
		\begin{itemize}
			\item [(i)]	$\exp\bigg({i\frac{g.\theta}{2}}\bigg)=\dfrac{\alpha \exp({i\frac{\theta}{2}})+\beta \exp({-i\frac{\theta}{2}})}{\mid\alpha \exp({i\frac{\theta}{2}})+\beta \exp({-i\frac{\theta}{2}})\mid};$
			\item [(ii)]$\exp({t\left(g,\theta\right)})=\mid\alpha \exp({i\theta})+\beta\mid^{2}=\mid\bar{\alpha}+\bar\beta \exp({i\theta})\mid^{2};$
			
			\item[(iii)]$d\dfrac{\left(g,\theta\right)}{d\theta}=\mid\bar{\alpha}+\bar\beta \exp({i\theta})\mid^{-2}=\exp({-t\left(g,\theta\right)}).$
		\end{itemize}
	\end{enumerate}
\end{thm}  
\begin{proof}[Proof of (1)]\cite{takahashi1961fonctions}
	We begin by showing that the $\Spin_{\BR}(\frac{1}{2})$ is isomorphic to the $\SU(1,1)$-quasi boson. Let $\mathcal{D}=\left\{ z\in\mathbb{C}:\mid z\mid<1\right\} $ be the unit disc. The Cayley transformation
	\[
	c:\thinspace z\longmapsto c(z)=\frac{z-i}{z+i}
	\]	
	transforms the upper half-plane $P$ onto $\mathcal{D}.$ We can
	now speak of an isomorphism from $A(P)$ onto $A(\mathcal{D})$ given by
	\begin{equation*}
		w:g\longmapsto Cg C^{-1}.
	\end{equation*}
	Let C be a 2$\times$2 Unitary matrix given by 
	\begin{equation*}
		C=\frac{1}{\sqrt{2}}\left(\begin{array}{cc}
			1 & -i\\
			1 & i
		\end{array}\right)\in U\left(2\right),
	\end{equation*}
	and let $g=\left(\begin{array}{cc}
		a & b\\
		c & d
	\end{array}\right)\in M_{2}\left(\mathbb{R}\right)$ . Then 
	
	\begin{align*}
	Cg C^{-1}&=\frac{1}{2}\left(\begin{array}{cc}
		1 & -i\\
		1 & i
	\end{array}\right)\left(\begin{array}{cc}
		a & b\\
		c & d
	\end{array}\right)\left(\begin{array}{cc}
		1 & 1\\
		i & -i
	\end{array}\right)\\
		&=\frac{1}{2}\left[\begin{array}{cc}
			\left(a+d\right)+i\left(b-c\right) & \left(a-d\right)-i\left(b+c\right)\\
			\left(a-d\right)+i\left(b+c\right) & \left(a+d\right)-i\left(b-c\right)
		\end{array}\right].
	\end{align*}
	Thus $C\cdotp g\cdotp C^{-1}$ is of the
	form $\left(\begin{array}{cc}
		\alpha & \beta\\
		\bar{\beta} & \bar{\alpha}
	\end{array}\right)$
	where 
	
	\[
	\alpha=\frac{1}{2}\left\{ \left(a+d\right)+i\left(b-c\right)\right\} 
	\]
	\[
	\beta=\frac{1}{2}\left\{ \left(a-d\right)-i\left(b+c\right)\right\} 
	\]
	with $\det g = \mid\alpha^{2}\mid-\mid\beta^{2}\mid= 1$ where $\alpha,\beta\in\mathbb{C}$.
	It is a subgroup of $GL\left(2,\mathbb{\mathbb{C}}\right)$ formed
	from matrix $g$ such that
	\begin{equation}
		\mathrm{G}=\SU\left(1,1\right)=\left\{ g\in GL(2,\mathbb{\mathbb{C}})\mid\bar{g}^{t}\sigma_{z}\thickspace g=\sigma_{z},\thinspace\det g=1\right\} 
		\label{eq115}
	\end{equation}
	where $\sigma_{z}=\left(\begin{array}{cc}
		1 & 0\\
		0 & -1
	\end{array}\right)$, and $\bar{g}^{t}$ is the adjoint matrix of $g$. 
\end{proof}	
\begin{proof}[poof of (2)]
	From equation \eqref{eq115} and by the isomorphism $w:g\longmapsto Cg C^{-1}$ from $\mathrm{G}\longmapsto SU\left(1,1\right)$, we have:
	\begin{align*}
		k_{\theta}&=w( k_{\theta})=\frac{1}{2}\left(\begin{array}{cc}
			1 & -i\\
			1 & i
		\end{array}\right)\left(\begin{array}{cc}
			\cos\frac{\theta}{2} & \sin\frac{\theta}{2}\\
			-\sin\frac{\theta}{2} & \cos\frac{\theta}{2}
		\end{array}\right)\left(\begin{array}{cc}
			1 & 1\\
			i & -i
		\end{array}\right)\\ &=\left(\begin{array}{cc}
			\exp({i\frac{\theta}{2}}) & 0\\
			0 & \exp({-i\frac{\theta}{2}})
		\end{array}\right),
	\end{align*}
	from similar and direct computation we obtain;
	\begin{equation*}
		d^{\frac{1}{2}}_{t}=w\left( d^{\frac{1}{2}}_{t}\right)= \left(\begin{array}{cc}
			\cosh\frac{t}{2} & \sinh\frac{t}{2}\\
			\sinh\frac{t}{2} & \cosh\frac{t}{2}
		\end{array}\right),
		\label{eq117}
	\end{equation*}
	\begin{equation*}
		n_{\xi}=w\left( n_{\xi}\right)= \left(\begin{array}{cc}
			1+i\frac{\xi}{2} & -i\frac{\xi}{2}\\
			i\frac{\xi}{2} & 1-i\frac{\xi}{2}
		\end{array}\right).
		\label{eq118}
	\end{equation*}
	Now we have the following consequence;
	\begin{equation*}
		a-ic=\alpha+\beta,
	\end{equation*}	
	we obtain
	\begin{gather}
		\exp\bigg({i\frac{\theta}{2}}\bigg)=\frac{\alpha+\beta}{ \mid\alpha+\beta\mid},\\ \nonumber
		\exp({t})=\mid\alpha+\beta\mid^{2},\\ \nonumber
		\xi=Im\left(\frac{\alpha-\beta}{\alpha+\beta}\right).
		\label{116}
	\end{gather}
	We have $
	g= \mbox{\CYRZH}  k_{\theta}d^{\frac12}_{t}n_{\xi}
	$ and when we set $\mbox{\CYRZH}=1$ one gets the usual Iwasawa decomposition. 
\end{proof}	
\begin{proof}[Proof of (3)]
	Also applying isomorphism $w:g\longmapsto Cg C^{-1}$ we
	have the following:
	\begin{equation*}
		w\left(g\cdot k_{\theta}\right)=k_{g.\theta}d^{\frac{1}{2}}_{t\left(g,\theta\right)}n_{\xi\left(g,\theta\right)},
		\label{eq123}
	\end{equation*}
	
	if $g=\left(\begin{array}{cc}
		\alpha & \beta\\
		\bar{\beta} & \bar{\alpha}
	\end{array}\right)\in \SU\left(1,1\right)$, 
	then \begin{gather*}
		g k_{\theta}=\left(\begin{array}{cc}
			\alpha & \beta\\
			\bar{\beta} & \bar{\alpha}
		\end{array}\right)\left(\begin{array}{cc}
			\exp({i\frac{\theta}{2}}) & 0\\
			0 & \exp({-i\frac{\theta}{2}})
		\end{array}\right)=\left(\begin{array}{cc}
			\alpha \exp({i\frac{\theta}{2}}) & \beta \exp({-i\frac{\theta}{2}})\\
			\bar{\beta}\exp({i\frac{\theta}{2}}) & \bar{\alpha}\exp({-i\frac{\theta}{2}})
		\end{array}\right).
	\end{gather*} 
	We can apply equation (\eqref{116}) to $\alpha\exp({i\frac{\theta}{2}})$ and $\beta \exp({-i\frac{\theta}{2}})$ to obtain Theorem \eqref{thm2.2} [3(i), (iii), and(iii)].
	This completes the proof.	
\end{proof}

\begin{lem}\label{lem2.3}\cite{takahashi1961fonctions, hounkonnou2022group}
	Any element $g$ in $\SU\left(1,1\right)-$quasi boson or $\Spin (\frac{1}{2})$ can be written as 
	\begin{align*}
		g&=k_{\varphi}d^{\frac{1}{2}}_{t}k_{\psi}\\
	&=\exp\left(\varphi\bra{s,m}{\sigma_{k}}\ket{s,m}\right) 
		\exp\left(t\bra{s,m}{\sigma_{z}}\ket{s,m}\right) 
		\exp\left(\psi\bra{s,m}{\sigma_{k}}\ket{s,m}\right).
	\end{align*}	
	for $0\leq\varphi\leq4\pi$ , $0\leq t$ , $0\leq\psi\leq2\pi$ if
	$g=\left(\begin{array}{cc}
		\alpha & \beta\\
		\bar{\beta} & \bar{\alpha}
	\end{array}\right)\in \mathrm{G}$ , then 
	
	\begin{enumerate}
		\item (i) $\alpha=\exp({i\frac{(\varphi+\psi)}{2}})\cosh\frac{t}{2}$,
		\\(ii) $ \beta=\exp({i\frac{(\varphi-\psi)}{2}})\sinh\frac{t}{2}$,
		\item (i) $\sinh\frac{t}{2}=\mid\beta\mid$,
		\\ (ii) $\cosh\frac{t}{2}=\left(1+\mid\beta\mid^{2}\right)^{\frac{1}{2}}=\mid\alpha\mid$,\\ 
		(iii) $\exp({\frac{t}{2}})=\mid\alpha\mid+\mid\beta\mid$,
		\\ 
		In particular if $g$ belongs to $\mathrm{G}-K$, then 
		
		\item (i)
		$\exp({i\frac{\left(\varphi+\psi\right)}{2}})=\frac{\alpha}{\mid\alpha\mid}$,\\
		(ii) $\exp({i\frac{\left(\varphi-\psi\right)}{2}})=\frac{\beta}{\mid\beta\mid}$,\\
		(iii) $\exp({i\varphi})=  \frac{\left(\alpha\beta\right)}{\mid\alpha\beta\mid}$, \\
		(iv) $\exp({i\psi})=  \frac{\left(\alpha\beta^{-1}\right)}{\mid\alpha\beta^{-1}\mid}$, \\
		and 
		$\left(\varphi, t,\psi\right)$ is uniquely
		determined by $g.$ If $g\in K$, then $t=0$ and $\varphi+\psi$
		is determined modulo $4\pi$ by $g$. 
	\end{enumerate}
\end{lem}

\begin{prop}\label{2.4}\cite{takahashi1961fonctions}
	\begin{enumerate}
	\item[(1)] The left-invariant Haar integration for any element $g\in \Spin(\frac{1}{2})$ or $\SU\left(1,1\right)$-quasi boson of $G$ of a spin particle is given by
	\begin{equation*}
		\int_{\mathrm{G}}f\left(g\right)dg=\frac{\mbox{\CYRZH}}{4\pi}\int_{0}^{4\pi}\int_{-\infty}^{\infty}\int_{-\infty}^{\infty}f\left(k_{\theta}d^{\frac{1}{2}}_{t}n_{\xi}\right)\exp({t})d\theta dtd\xi
	\end{equation*}
	for any continuous function with compact support, where $g=k_{\theta}d^{\frac{1}{2}}_{t}n_{\xi}$
	is the Iwasawa decomposition by setting \CYRZHDSC=1. 
	
	\item[(2)] In the case of an electron we have that $g=\hbar k_{\theta}d^{\frac{1}{2}}_{t}n_{\xi}$, where $\hbar=\dfrac{h}{2\pi}$, thus we have the left-invariant Haar measure to be
	\begin{align}
		\int_{\mathrm{G}}f\left(g\right)dg=\frac{\hbar}{4\pi}\int_{0}^{4\pi}\int_{-\infty}^{\infty}\int_{-\infty}^{\infty}f\left(k_{\theta}d^{\frac{1}{2}}_{t}n_{\xi}\right)\exp({t})d\theta dtd\xi \\ \nonumber =\frac{h}{8\pi^{2}}\int_{0}^{4\pi}\int_{-\infty}^{\infty}\int_{-\infty}^{\infty}f\left(k_{\theta}d^{\frac{1}{2}}_{t}n_{\xi}\right)\exp({t})d\theta dtd\xi
	\end{align}
	\end{enumerate}
\end{prop}
\begin{proof}
	Put $dg=\exp({t})d\theta dtd\xi$, we shall prove that $d\left(g_{o}g\right)=dg$
	for any $g_{o}\in \mathrm{G}$. Every element of $g_{o}\in \mathrm{G}$ can be written
	as 
	\begin{equation*}
		g_{o}=k_{\varphi}d^{\frac{1}{2}}_{\tau}k_{\psi}
	\end{equation*}
	hence it is sufficient to prove that $d\left(g_{o}g\right)=dg$ for
	$g_{o}=k_{\varphi}$ and $d^{\frac{1}{2}}_{\tau}$ . Since 
	\begin{equation*}
		k_{\varphi}k_{\theta}d^{\frac{1}{2}}_{t}n_{\xi}=k_{\varphi+\theta}d^{\frac{1}{2}}_{t}n_{\xi}
	\end{equation*}	
	we obtain 
	\begin{equation*}
		d\left(k_{\varphi}g\right)=\exp({t})d\left(\varphi+\theta\right)dtd\xi=\exp({t})d\theta dtd\xi=dg
	\end{equation*}
	By simple computation we have 
	\begin{equation*}
		d^{\frac{1}{2}}_{t}n_{\xi}d^{\frac{1}{2}}_{t^{-1}}=n_{\exp({t})\xi}
	\end{equation*}
	we put $t^{\prime}=t\left(d^{\frac{1}{2}}_{\tau},\theta\right)$ and $\xi^{\prime}=\xi\left(d^{\frac{1}{2}}_{\tau},\theta\right)$. Then we have 
	\begin{align*}
	d^{\frac{1}{2}}_{\tau}k_{\theta}d^{\frac{1}{2}}_{\tau}n_{\xi}&=k_{d^{\frac{1}{2}}_{\tau}.\theta}d^{\frac{1}{2}}_{t^{\prime}}n_{\xi^{\prime}}d^{\frac{1}{2}}_{t}n_{\xi}\\
	&=k_{d^{\frac{1}{2}}_{\tau}.\theta}d^{\frac{1}{2}}_{t^{\prime}+t}n_{\exp({-t})\xi^{\prime}+\xi}\\
	&=k_{d^{\frac{1}{2}}_{\tau}.\theta}d^{\frac{1}{2}}_{t^{\prime}+t}n_{\exp({-t})\xi^{\prime}+\xi}.
\end{align*}
	Hence we have
	\begin{align*}
		d\left(d^{\frac{1}{2}}_{\tau}g\right)&=  \exp({t+t^{\prime}})d\left(d^{\frac{1}{2}}_{\tau}.\theta\right)d\left(t+t^{\prime}\right)d\left(\xi+\exp({-t})\xi^{\prime}\right)\\
		&=\exp({t+t^{\prime}})\exp({-t^{\prime}})d\theta dtd\xi=dg.
	\end{align*}
The proof of [(2)] is straight forward.
Thus the proof is complete.
\end{proof}

\begin{prop}
	The Haar integration proposition \eqref{2.4} is given by 
	\begin{equation*}
		\int_{\mathrm{G}}f\left(g\right)dg=2\pi\int_{K}\int_{0}^{\infty}\int_{K}f\left(kd^{\frac{1}{2}}_{t}k^{\prime}\right)\sinh tdkdtdk^{\prime}
	\end{equation*}
and for an electron we obtain;
	\begin{align*}
	\int_{\mathrm{G}}f\left(g\right)dg&=2\pi\hbar\int_{K}\int_{0}^{\infty}\int_{K}f\left(kd^{\frac{1}{2}}_{t}k^{\prime}\right)\sinh tdkdtdk^{\prime}
	\\
	&=2\pi \left(\dfrac{h}{2\pi}\right)\int_{K}\int_{0}^{\infty}\int_{K}f\left(kd^{\frac{1}{2}}_{t}k^{\prime}\right)\sinh tdkdtdk^{\prime}\\
	&=h\int_{K}\int_{0}^{\infty}\int_{K}f\left(kd^{\frac{1}{2}}_{t}k^{\prime}\right)\sinh tdkdtdk^{\prime}
\end{align*}
	where $dk$ is the normalized Haar measure $\left(\frac{1}{4\pi}\right)d\theta$
	of $K=\left\{ k_{\theta}\mid0\leq\theta<4\pi\right\} .$ 
\end{prop}
\begin{proof}
	Let $g=k_{\varphi}d^{\frac{1}{2}}_{\tau}k_{\psi}=k_{\theta}d^{\frac{1}{2}}_{t}n_{\xi}$ as in \cite{hounkonnou2022group}, where \CYRZHDSC=1, 
	$0\leq\theta<4\pi$, $0\leq\varphi$, $0\leq\tau<+\infty$, $-\infty\leq t,\xi<+\infty$
	and $0\leq\psi<2\pi$ .
	To demonstrate the integration formula, we have to calculate the
	Jacobian from the change of variable in $G-K$, 
	$\left(\theta,\thinspace t,\thinspace\xi\right)\longmapsto\left(\varphi,\thinspace\tau,\thinspace\psi\right)$ by putting 	
	\begin{equation}
		\left(\begin{array}{cc}
			\alpha & \beta\\
			\bar{\beta} & \bar{\alpha}
		\end{array}\right)=k_{\theta}d^{\frac{1}{2}}_{t}n_{\xi}=k_{\varphi}d^{\frac{1}{2}}_{\tau}k_{\psi}
	\end{equation}
	we have
	\begin{gather*}
		\alpha=\exp\bigg({i\frac{(\varphi+\psi)}{2}}\bigg)\cosh\frac{\tau}{2}=\exp\bigg({i\frac{\theta}{2}}\bigg)\left(\cosh\frac{t}{2}+i\frac{\xi}{2}\exp\bigg({\frac{t}{2}}\bigg)\right),
		\\ \nonumber
		\beta=\exp\bigg({i\frac{(\varphi-\psi)}{2}}\bigg)\sinh\frac{\tau}{2}=\exp\bigg({i\frac{\theta}{2}}\bigg)\left(\sinh\frac{t}{2}-i\frac{\xi}{2}\exp\bigg({\frac{t}{2}}\bigg)\right),
	\end{gather*}	
	and 
	\begin{equation*}
		2\alpha\beta=\exp({i\varphi})\sinh\tau=\exp({i\theta})\left(\sinh t+\frac{\xi^{2}}{2}\exp({t})-i\xi\right)
	\end{equation*}
	from which we obtain
	\begin{equation*}
		\sinh^{2}\tau=\left(\sinh t+\frac{\xi^{2}}{2}\exp({t})\right)^{2}+\xi^{2}
	\end{equation*}
	and 
	\begin{equation*}
		\exp({i\varphi})=\exp({i\theta})\frac{\sinh t+\frac{\xi^{2}}{2}\exp({t})-i\xi}{\sqrt{(\sinh t+\frac{\xi^{2}}{2}\exp({t}))^{2}+\xi^{2}}}
	\end{equation*}
	thus we have $\frac{\partial\tau}{\partial\theta}=0$ , $\frac{\partial\varphi}{\partial\theta}=1$, 
	moreover, we have
	\begin{equation}
		\frac{\beta}{\alpha}=\exp({-i\psi})\tanh\frac{\tau}{2}=\frac{\sinh\frac{t}{2}-i\frac{\xi}{2}\exp({\frac{t}{2}})}{\cosh\frac{t}{2}+i\frac{\xi}{2}\exp({\frac{t}{2}})}
		\label{eq152}
	\end{equation}	
	this shows that both $\psi$ and $\tau$ do not depend on $\theta.$
	We therefore have $\frac{\partial\psi}{\partial\theta}=0$. 
	Differentiating equation \eqref{eq152} with respect to $t$, we
	obtain
	\begin{equation*}
		-i\exp({-i\psi})\tanh\frac{\tau}{2}\frac{\partial\psi}{\partial t}+\frac{\exp({-i\psi})}{2\cosh^{2}\frac{\tau}{2}}\frac{\partial\tau}{\partial t}=\frac{1}{2\left(\cosh\frac{t}{2}+i\frac{\xi}{2}\exp({\frac{t}{2}})\right)^{2}}
	\end{equation*}	
	where 
	\[
	\exp({-i\psi})\tanh\frac{\tau}{2}\left(-i\frac{\partial\psi}{\partial t}+\frac{1}{\sinh\tau}\frac{\partial\tau}{\partial t}\right)=\frac{1}{2\left(\cosh\frac{t}{2}+i\frac{\xi}{2}\exp({\frac{t}{2}})\right)^{2}}
	\]	
	from which we obtain
	\begin{align*}
		-i\frac{\partial\psi}{\partial t}+\frac{1}{\sinh\tau}\frac{\partial\tau}{\partial t}&=\frac{1}{2\left(\cosh\frac{t}{2}+i\frac{\xi}{2}\exp({\frac{t}{2}})\right)\left(\sinh\frac{t}{2}-i\frac{\xi}{2}\exp({\frac{t}{2}})\right)}\\  &=\frac{1}{2\alpha\beta \exp({-i\theta})}=\frac{\exp({i\left(\theta-\varphi\right)})}{\sinh\tau}
	\end{align*}
	therefore, we have $\frac{\partial\psi}{\partial t}=\frac{-\sin\left(\theta-\varphi\right)}{\sinh\tau}$,
	$\frac{\partial\tau}{\partial t}=\cos\left(\theta-\varphi\right)$.
	Similarly, we can find 
	\begin{equation*}
		-i\frac{\partial\psi}{\partial\xi}+\frac{1}{\sinh\tau}\frac{\partial\tau}{\partial\xi}=-i\exp({t})\frac{\exp({i\left(\theta-\varphi\right)})}{\sinh\tau},
	\end{equation*}
	which gives $\frac{\partial\psi}{\partial\xi}=\frac{\exp({t})\cos\left(\theta-\varphi\right)}{\sinh\tau}$,
	$\frac{\partial\tau}{\partial\xi}=\exp({t})\sin\left(\theta-\varphi\right)$
	from which we deduce that 
	\begin{equation*}
		\bigg|\frac{\partial\left(\varphi,\thinspace\tau,\thinspace\psi\right)}{\partial\left(\theta,\thinspace t,\thinspace\xi\right)}\bigg|=\frac{\exp({t})}{\sinh\tau}
	\end{equation*}
	thus by transforming the Haar integral in proposition \eqref{2.4} we
	have
	\begin{align*}
		\int_{\mathrm{G}}f\left(g\right)dg&=\int_{k}\int_{-\infty}^{\infty}\int_{-\infty}^{\infty}f\left(kd^{\frac{1}{2}}_{\tau}n_{\xi}\right)\exp({\tau})dkdtd\xi\\
		&=\int_{k}\int_{0}^{\infty}\int_{0}^{2\pi}f\left(kk_{\theta}^{-1}d^{\frac{1}{2}}_{t}k_{\psi}\right)\sinh tdkdtd\psi\\
		&=\int_{k}\int_{0}^{\infty}\int_{0}^{2\pi}f\left(kd^{\frac{1}{2}}_{t}k_{\psi+2\pi}\right)\sinh tdkdtd\psi\\
		&=\frac{1}{2}\int_{k}\int_{0}^{\infty}\int_{0}^{4\pi}f\left(kd^{\frac{1}{2}}_{t}k_{\psi}\right)\sinh tdkdtd\psi\\
		&=2\pi\int_{K}\int_{0}^{\infty}\int_{0}^{4\pi}f\left(kd^{\frac{1}{2}}_{t}k^{\prime}\right)\sinh tdkdtdk^{\prime}
	\end{align*}
	using the relation $d^{\frac{1}{2}}_{t}k_{\psi}=k_{2\pi}d^{\frac{1}{2}}_{\tau}k_{\psi+2\pi}$ and $\hbar \cdotp kd^{\frac{1}{2}}_{t}k^{\prime}=\dfrac{h}{2\pi} kd^{\frac{1}{2}}_{t}k^{\prime}$, one can easily obtain

    \begin{align*}
\int_{\mathrm{G}}f\left(g\right)dg=h\int_{K}\int_{0}^{\infty}\int_{K}f\left(kd^{\frac{1}{2}}_{t}k^{\prime}\right)\sinh tdkdtdk^{\prime}.
\end{align*}
    
\end{proof}

\begin{thm}\label{2.7}\cite{takahashi1961fonctions}
Let $S_{-}$, $S_{+}$, $S_{z}$ be basis of the $\spin_{\BR}(\frac12)$ Lie algebra with $\hbar=1$, then the following results hold:
	\begin{enumerate}
		\item[(1)] For any $\varphi\in \mathcal{H_{\infty}}(L^{2}(U)_{\infty})$ we have
		\begin{align*}
			dV^{j,s}_{S_x}\varphi=&\frac{1}{2}((s-j)\zeta+(s+j)\zeta^{-1})\varphi-\frac{i}{2}(\zeta-\zeta^{-1})(\mathbb{D}\varphi)\zeta,\\ 
			dV^{j,s}_{S_y}\varphi=&\frac{i}{2}((s-j)\zeta-(s+j)\zeta^{-1})\varphi+\frac{1}{2}(\zeta+\zeta^{-1})(\mathbb{D}\varphi)\zeta,
		\end{align*}
		where $\mathbb{D}$ denotes the differential operator given by
		\begin{gather*}
			\mathbb{D}\varphi(\zeta)=\underset{t \longrightarrow 0}\lim t^{-1}[\varphi(\exp(it(\zeta)-\varphi(\zeta)].
		\end{gather*}
		\item[(2)] Let $\varphi_{p}(\zeta)=(\zeta)^{-p}, p\in\mathbb{Z}$. Then
		\begin{equation*}
			\left\lbrace\begin{array}{ll}
				dV^{j,s}_{S_x}\varphi=\frac{s-p-j}{2}\varphi_{p-1}+\frac{s+p+j}{2}\varphi_{p+1},\\
				dV^{j,s}_{S_y}\varphi=i\frac{s-p-j}{2}\varphi_{p-1}-i\frac{s+p+j}{2}\varphi_{p+1}.
			\end{array}
			\right.
		\end{equation*}
		\item[(3)] 	Given that
		$S_{+}=S_x+iS_y$ and $S_{-}=S_x-iS_y$ are the ladder operators of a spin particle with real Lie algebra $g$, for $S \in g (\spin_{\BR}(\frac12))$, we can extend the definition of $dV^{j,s}_{S}$ to the ladder operators, $S_{+}$ and $S_{-}$, formally by
		\begin{gather*}
			dV^{j,s}_{S_{\pm}}=dV^{j,s}_{S_{x}} \pm i dV^{j,s}_{S_{y}}, \quad where \quad  S_{\pm}=S_x\pm iS_y, \quad and \quad  S_{x},S_{y}\in g.
		\end{gather*}
		Then we have
		\begin{equation*}
			dV^{j,s}_{S_{+}}\varphi_p=(s+p+j)\varphi_{p+1}, 
		\end{equation*}
		and
		\begin{equation*}
			dV^{j,s}_{S_{-}}\varphi_p=(s-p-j)\varphi_{p-1}.
		\end{equation*}
	\end{enumerate}
\end{thm}
\begin{proof}[Proof of (1)]\cite{takahashi1961fonctions}
	\begin{enumerate}
		\item [] Let $d^{\frac{1}{2}}_{t}=\exp t(S_{x})$, then we obtain
		\begin{align}\label{4.18}
			\left(V^{j,s}_{d_{t}^{\frac{1}{2}}}\varphi \right)(\zeta)=	\exp({-st(d_{-t}^{\frac{1}{2}},\zeta)})u(d_{-t}^{\frac{1}{2}},\zeta)^{2j} \varphi(d_{-t}^{\frac{1}{2}}\cdot\zeta)
		\end{align}
		and 
		\begin{align*}
			d^{\frac{1}{2}}_{t}=\exp t(S_x)=\begin{pmatrix}
				\cosh \frac{t}{2} & \sinh \frac{t}{2}\\ \sinh \frac{t}{2} &\cosh \frac{t}{2}
			\end{pmatrix}.
		\end{align*}
		We shall differentiate equation (\ref{4.18}) with respect to $t$ and set $t=0$. First we look at the product of 
		\begin{align}\label{4.19}
			\exp({-st(d_{-t}^{\frac{1}{2}},\zeta)})u(d_{-t}^{\frac{1}{2}},\zeta)^{2j}.
		\end{align}
		From $$\exp({t(g,\zeta)})=|\bar{\beta}\zeta+\bar{\alpha}|^2=|{\alpha}\zeta+{\beta}|^2$$ with $\zeta=\exp(i\theta)$. We put 
		\begin{align}
			\exp\left({t\left(d^{\frac{1}{2}}_{t},\zeta\right)}\right) \nonumber &=\mid\alpha \zeta+\beta\mid^{2}=|\bar{\beta}\zeta+\bar{\alpha}|^2  \\ \nonumber &=\bigg|\zeta\sinh\frac{t}{2}+\cosh\frac{t}{2}\bigg|^{2}.
		\end{align}
		When we put
		\begin{gather}\label{4.20}
			d_{-t}^{\frac{1}{2}}=g^{-1}= \begin{pmatrix}
				\bar{\alpha} & -\beta \\ \ -\bar{\beta} &\alpha
			\end{pmatrix}=\begin{pmatrix}
				\cosh \frac{t}{2} & -\sinh \frac{t}{2}\\ -\sinh \frac{t}{2} &\cosh \frac{t}{2}
			\end{pmatrix},
		\end{gather}
		we obtain
		\begin{gather}
			\exp\left({t\left(d^{\frac{1}{2}}_{-t},\zeta\right)}\right) \nonumber =\mid\bar{\beta}\zeta+\bar{\alpha}\mid^2  =\bigg|-\zeta\sinh\frac{t}{2}+\cosh\frac{t}{2}\bigg|^{2} \\
			\dfrac{\exp\left({t\left(d^{\frac{1}{2}}_{-t},\zeta\right)}\right)}{2} \nonumber  =\bigg|-\zeta\sinh\frac{t}{2}+\cosh\frac{t}{2}\bigg|.
		\end{gather}
		For	$$u(g,\zeta)=\dfrac{\bar{\beta} \zeta+\bar{\alpha}}{|\bar{\beta} \zeta+\bar{\alpha}|}=\dfrac{\bar{\beta} \zeta+\bar{\alpha}}{\exp({t(g,\zeta)})},$$
		similarly, 
		$$u(g^{-1},\zeta)=\dfrac{\bar{\beta} \zeta+\bar{\alpha}}{|\bar{\beta} \zeta+\bar{\alpha}|}=\dfrac{\bar{\beta} \zeta+\bar{\alpha}}{\exp({t(g^{-1},\zeta)})}.$$ Now from equation (\ref{4.20})
		\begin{gather*}
			u\left(d^{\frac{1}{2}}_{-t},\zeta\right)=\left(\dfrac{\exp\left({t\left(d^{\frac{1}{2}}_{-t},\zeta\right)}\right)}{2}\right)^{-1}(\bar{\beta} \zeta+\bar{\alpha})
		\end{gather*}
		similarly from equation (\ref{4.19}) we have:
		\begin{align*}
			(\exp({st(d_{-t}^{\frac{1}{2}},\zeta)}))^{-1}u(d_{-t}^{\frac{1}{2}},\zeta)^{2j}&=\exp((-(s+j))({t(d_{-t}^{\frac{1}{2}},\zeta)})) \\ \nonumber &\times
			\left(-\zeta\sinh\frac{t}{2}+\cosh\frac{t}{2}\right)^{2j}.
		\end{align*}
		We shall now differentiate term by term the following equation:
		\begin{align}
			\exp((-(s+j))({t(d_{-t}^{\frac{1}{2}},\zeta)})) \left(-\zeta\sinh\frac{t}{2}+\cosh\frac{t}{2}\right)^{2j}
		\end{align}
		\begin{align}
			\exp\left({t\left(d^{\frac{1}{2}}_{-t},\zeta\right)}\right) \nonumber  =\bigg|-\zeta\sinh\frac{t}{2}+\cosh\frac{t}{2}\bigg|^{2}=\cosh\frac{t}{2}-\bigg(\dfrac{\zeta+\bar{\zeta}}{2}\bigg)\sinh\dfrac{t}{2}.
		\end{align}
		Hence 
		\begin{align*}
			\dfrac{d}{dt}\exp\bigg(t\left(d^{\frac{1}{2}}_{-t},\zeta\right)\bigg)\bigg|_{t=0}=-\dfrac{\zeta+\zeta^{-1}}{2}.
		\end{align*}
		We also obtain
		\begin{gather*}
			\dfrac{d}{dt}\left(-\zeta\sinh\frac{t}{2}+\cosh\frac{t}{2}\right)\bigg|_{t=0}=-\dfrac{\zeta}{2}.
		\end{gather*}
		To deal with the third factor of the right member of equation (\ref{4.18}), we first calculate
		\begin{gather*}\label{4.24}
			\dfrac{d}{dt}\left(d_{-t}^{\frac{1}{2}}\cdot\zeta\right)\bigg|_{t=0}=\dfrac{d}{dt}\bigg(\dfrac{\zeta\cosh\frac{t}{2}-\sinh\frac{t}{2}}{-\zeta\sinh\frac{t}{2}+\cosh\frac{t}{2})}\bigg)\bigg|_{t=0}=\dfrac{\zeta^{2}-1}{2}.
		\end{gather*}
	\end{enumerate}	
	If we set $d_{-t}^{\frac{1}{2}}\cdot\zeta=\exp (i\theta(t))$ and write $\varphi(d_{-t}^{\frac{1}{2}}\cdot\zeta)=f(\theta(t))$, then we have 
	\begin{align*}
		\dfrac{d}{dt}\varphi(d_{-t}^{\frac{1}{2}}\cdot\zeta)\bigg|_{t=0}=f^{\prime}(\theta(0))\theta^{\prime}(0).
	\end{align*}
	Note that;
	\begin{gather*}
		f^{\prime}(\theta(0))=\mathbb{D}\varphi(\zeta)=\underset{t \longrightarrow 0}\lim t^{-1}[\varphi(\exp(it(\zeta)-\varphi(\zeta)].  
	\end{gather*}
	We obtain $\theta^{\prime}(0)$ from $$\dfrac{d}{dt}(d_{-t}^{\frac{1}{2}}\cdot\zeta)|_{t=0}=\exp (i\theta(t))i\theta^{\prime}(0)=i\zeta\theta^{\prime}(0),$$ which, combined with equation (\ref{4.24}), yields $\theta^{\prime}(0)=\dfrac{\zeta-\zeta^{-1}}{2i}$. Putting these together, we obtain: 
	\begin{gather*}
		(dV_{S_{x}}^{j,s}\varphi)(\zeta)=(j+s)\dfrac{\zeta+\zeta^{-1}}{2}\varphi+2j \left(-\dfrac{\zeta}{2}\right)\varphi(\zeta)-\dfrac{i}{2}(\zeta-\zeta^{-1})(\mathbb{D}\varphi)\zeta,
	\end{gather*}
	giving the first formula in Theorem \ref{2.7} $(1).$ If we replace the $d_{-t}^{\frac{1}{2}}$ above by 
	\begin{equation*}
		g_{t}=\exp t(S_y)=\begin{pmatrix}
			\cosh \frac{t}{2} & -i\sinh \frac{t}{2}\\ i\sinh \frac{t}{2} &\cosh \frac{t}{2}
		\end{pmatrix}
	\end{equation*}
	and make the corresponding changes all the way through, then we arrive at the second formula of Theorem \ref{2.7} $(1)$.
\end{proof}
\begin{proof}[Proof of 2]
	To do this we first apply Theorem \eqref{2.7}(1) to $\varphi=\varphi_{p}$ and make use of the following;
	$$\zeta \varphi_{p}=\varphi_{p-1}, \quad \zeta^{-1}\varphi_{p}=\varphi_{p+1}, \quad \mathbb{D}\varphi_{p}=-ip\varphi_{p}.$$
\end{proof}
\begin{proof}[Proof of 3]
	This is a simple computation by applying Theorem\eqref{2.7} (1),(2).
\end{proof}

\section{Application of Kangni Spherical Fourier Transform of the Type Delta to Spin particles	}\label{SFT}
In this section we shall follow strictly the work of the second author \cite{kangni2001transformation} and extend the results to the spin particle  and the quasi-boson.
\begin{thm}\cite{kangni2001transformation}
	For any element $g\in \SU(1,1)$-quasi boson spin Lie group and $\mu$  a linear form on $A$. The spherical Fourier transformation of type $\chi_n$ on $\mathrm{G}$ defined by:
	\begin{align*}
		\phi^\mu _n (f)&= \frac{h}{8\pi^{2}} \int_G \int_0 ^{4\pi}   \left( \frac{-\beta \zeta+ \alpha}{|\alpha -\beta \zeta|} \right)^{2n}\ \  f(g)
		\mbox{exp} \ [\mu (log\  d_t^{\frac{1}{2}})] \ dg\  d\theta\\
		&\mbox{where}\quad \zeta=\exp(i\theta) \quad\mbox{and}\quad
		g=\begin{pmatrix}
			\alpha & \beta  \\
			\bar{\beta} & \bar{\alpha} \\
		\end{pmatrix}  \mbox{for all f}\in \mathcal{K}^\natural _n (\mathrm{G})).
	\end{align*}
When $\hbar=1$ we get the Kangni-type transform \eqref{KTT}.
\end{thm}

\begin{proof}\cite{kangni2001transformation}
	Let $f\mapsto F_f ^{<n>}$ be Abel's transformation generalised on
	$G$ following the class $\chi_n$. We have:
	$$F_f ^{<n>}(t) = \frac{\hbar \exp({t/2})}{4\pi} \int_0 ^{4\pi}  \int_{- \infty} ^{+ \infty} f(u_\theta d_t^{\frac{1}{2}} n_\xi) \exp({-in\theta}) d\xi d\theta.$$
	Thus the spherical Fourier's transformation of type $\chi_n$ is
	defined by:
	$$ \phi_n ^\mu = \int_{- \infty} ^{+ \infty} F_f ^{<n>}(t) \exp({t/2}) \mbox{exp}[\mu (log d_t^{\frac{1}{2}})]
	dt.$$  If $g= \hbar u_\theta d_t^{\frac{1}{2}} n_\xi $, then $dg=
	\frac{\hbar}{4\pi} \exp({t/2}) dt\  d\xi) .$ Therefore
	\begin{align*}
		\phi_n ^\mu (f) &= \frac{\hbar}{4\pi} \int_0 ^{4\pi} \int_{- \infty} ^{+ \infty} \int_{- \infty} ^{+ \infty} f(u_\theta d_t^{\frac{1}{2}} n_\xi) \chi_n(u_{-\theta}) \exp(t) exp[\mu(log d_t^{\frac{1}{2}})] d\theta\  dt\  d\xi \\
		&= \int_G f(g) \chi_n (K(g^{-1})) \mbox{exp} [\mu(log \ d_t^{\frac{1}{2}})] dg \\
		&= \frac{\hbar}{4\pi} \int_G \int_G ^{4\pi} f(g) \chi_n (K(u_{- \theta}g^{-1}u_\theta)) \mbox{exp} [\mu(log \ d_t^{\frac{1}{2}})] dg\  d\theta\\
		&= \frac{\hbar}{4\pi} \int_G \int_G ^{4\pi} f(g) \mbox{exp} [in(- \theta + g^{-1} \theta) + \mu(log \ d_t^{\frac{1}{2}})] dg\  d\theta\\
		&=\frac{\hbar}{4\pi} \int_G \int_G ^{4\pi} f(g) [\mbox{exp}
		(-ig^{-1}\cdotp\theta/2 + i\theta/2)]^{2n} \mbox{exp}[\mu(\log \ d_t^{\frac{1}{2}})] dg\
		d\theta
	\end{align*}
	We deduce from theorem \ref{thm2.2} that if $g=
	\begin{pmatrix}
		\alpha & \beta  \\
		\bar{\beta} & \bar{\alpha} \\
	\end{pmatrix} \in \SU(1,1)$-quasi boson we have:
	\begin{eqnarray}
		\exp({i(g.\theta)/2}) = \frac{\alpha \exp({i\theta/2}) + \beta
			\exp({-i\theta/2})}{|\alpha \exp({i\theta/2}) + \beta \exp({-i\theta/2})|}
	\end{eqnarray}
	
	Putting: $u(g, \zeta)= \frac{\beta\zeta +\alpha}{|\beta\zeta +\alpha|}$  we have $$\mbox{exp} [-ig^{-1}.\theta/2 +
	i\theta/2] = u(g^{-1}, \exp({i\theta})).$$ Hence
	$$ \phi_n ^\mu (f) = \frac{h}{8\pi^{2}} \int_G \int_G ^{4\pi} u (g^{-1}, e^{i\theta})^{2n} \mbox{exp}\left[\mu \left(\log \ d_ {t_{(g^{-1},\theta)}}^{\frac{1}{2}}\right)\right] dg\
	d\theta.$$
\end{proof}

\begin{thm}\cite{kangni2001transformation}
	Let $\omega \in G$ and $u_\alpha\ d_s^{\frac{1}{2}} \ u_\beta$ be as
in lemma\eqref{lem2.3}. For all function $f \in
	\mathcal{K}^\natural _n (G)$ we have:
	$$\phi^\mu _n (_\omega f) = \frac{\phi^\mu _n (f)}{\exp({in(\alpha +d_s^{\frac{1}{2}}\cdot \beta)+s/2})}\ ,$$
	where $_\omega f$ is defined for all $x \in G$ by $_\omega f(x) =
	f(\omega^{-1} x)$ and $d_s^{\frac{1}{2}}. \beta$ is the rotation angle  $s$ around the $y-$axis in the $\SU(1,1)$-quasi boson space $R(s)=\exp(-isK_{y})$ which is a fix choice of phase as in equation \ref{rota}, that is, the $d-$function given as $d_s^{\frac{1}{2}} u_\beta.$
\end{thm}

\begin{proof}
	Consider Abel's transformation generalised by $f\mapsto F_f ^{<n>}$
	following the class $\chi_n$. Let's show that $f \in \mathcal{K}_n
	^\natural(G)$ and $t \in G$ we have:
	$$ F_{\omega f}^{<n>} = \exp({s/2}) \chi_n \bigg(u_{-\alpha-d_s^{\frac{1}{2}}. \beta}\bigg) F_{f}^{<n>} (t) \ \ \ (\forall\  \omega = u_\alpha d_s^{\frac{1}{2}} u_\beta \in G.)$$
	
	\begin{align*}
		F_{\omega f}^{<n>}(t) &=\frac{\exp({t/2})}{4\pi} \int_{0}^{4\pi} \int_{- \infty}^{+ \infty} (_\omega f(x)f(u_{\theta} d_{t}^{\frac{1}{2}} n_\xi))
         \chi_n (u_\theta^{-1}) d\xi \ d\theta\\
	&=\frac{\exp({t/2})}{4\pi} \int_0 ^{4\pi} \int_{- \infty}^{+ \infty} f(\omega^{-1} u_\theta d_t^{\frac{1}{2}} n_\xi) \chi_n (u_{- \theta})d\xi \  d\theta\\
		&=\frac{\exp({t/2})}{4\pi} \int_{0}^{4\pi} \int_{- \infty}^{+ \infty}
		f( u_{-\beta} d_s^{-\frac{1}{2}} u_{-\alpha} u_\theta d_t^{\frac{1}{2}} n_\xi) \chi_{n} (u_{- \theta})d\xi \ d\theta
	\end{align*}
	By computation, we have:
	
	\begin{align*}
		u_{-\beta} d_s^{-\frac{1}{2}}u_{-\alpha} u_\theta d_t^{\frac{1}{2}} n_\xi &=  u_{d_s^{-\frac{1}{2}}\cdot (\theta-\alpha)-\beta} d_{t (d_s^{-\frac{1}{2}}, \theta-\beta)}^{\frac{1}{2}} n_{\xi(d_s^{-\frac{1}{2}},\theta-\beta)} d_t^{\frac{1}{2}} n_\xi\\
		&=  u_{\theta'}\  d_{t'}^{\frac{1}{2}}\  n_{\xi'}\  d_t^{\frac{1}{2}}\  n_\xi \\
		&=   u_{\theta'}\  d_{t'}^{\frac{1}{2}} \ d_{t'}^{\frac{1}{2}}\  d_t^{-\frac{1}{2}} \ n_{\xi'} d_t^{\frac{1}{2}}\
		n_\xi = u_{\theta'}\   d_{t'+t}^{\frac{1}{2}}\  n_{\xi' \exp({-1}) + \xi}
	\end{align*}
	with $ \theta' =  d_s^{-\frac{1}{2}}\cdot (\theta - \alpha)- \beta$; $t'= t(d_s^{-\frac{1}{2}}, \theta -
	\beta)$ and $\xi'= \xi (d_s^{-\frac{1}{2}}, \theta - \beta).$
Therefore
	\begin{align*}
		F_{\omega f}^{<n>} (t)&=\frac{\exp({t/2})}{4\pi} \int_0 ^{4\pi} \int_{- \infty} ^{+ \infty} f (u_{\theta'}  d_{t^{\prime}+t}^{\frac{1}{2}} n_{\xi'\exp({-t})+\xi}) \chi_n (u_{-\theta})\  d\xi \  d\theta\\
		&=\frac{\exp({t/2})}{4\pi} \int_0 ^{4\pi} \int_{- \infty} ^{+ \infty} f(u_\theta\  d_t^{\frac{1}{2}}\ n_\xi) \exp({-in[\alpha + d_s^{\frac{1}{2}} \theta +\beta]})\\
		&\times \exp({t(d_s^{-\frac{1}{2}},\gamma )}) d\xi \  d\theta\\
	\end{align*}
	
	with $\gamma = d_s^{\frac{1}{2}}\cdot(\theta + \beta) + \alpha.$ Since $d_s^{\frac{1}{2}} =
	\begin{pmatrix}
		\exp({s/2}) & 0  \\
		0 & \exp({-s/2}) \\
	\end{pmatrix}$ we obtain\\
	
	$$ \exp({t(d_s^{-\frac{1}{2}},\gamma)}) = \exp({s/2}) \mbox{\ and } \exp({-in\ d_s^{\frac{1}{2}}\cdot\theta}) =
	\exp({-s}) \exp({-in\theta})$$ (by theorem \ref{thm2.2} (3)) \\
	Thus:\\
	
	\begin{align*}
		F_{\omega }^{<n>} (t) &=\frac{\exp({t/2})}{4\pi}. \exp({-in(\alpha+d_s^{\frac{1}{2}}\cdotp \beta)}) \int_0 ^{4\pi} \int_{- \infty} ^{+ \infty} f ((u_{\theta} d_t^{\frac{1}{2}} n_\xi)\exp({-s/2}) )\\
		&\times \chi_n (u_{-\theta})\  d\xi \  d\theta\\
		&= \exp({-s/2}) \chi_n \bigg(u_{-\alpha-d_s^{\frac{1}{2}}. \beta}\bigg) F_{\omega}^{<n>} (t) \\
	\end{align*}
	we have
	
	\begin{align*}
		\phi_n ^\mu (_\omega f) &= \frac{1}{\exp({s/2})}
		\chi_n\bigg(u_{-\alpha-d_s^{\frac{1}{2}}. \beta}\bigg) \ \phi_n ^\mu (f) \\
		&= \frac{\phi_n ^\mu (f)}{\exp({in(\alpha + d_s^{\frac{1}{2}}\cdot\beta) + s/2})} 
	\end{align*}
\end{proof}

\section{Concluding Remarks}\label{rem}
In this paper, we proved that the $\Spin(\frac12)$ Lie group is isomorphic to the $\SU(1,1)$-quasi boson. The universal enveloping algebra for the $\spin(\frac12)$ is developped and we showed that this is the same as the quantum $\sl_{q}(2,\BR)$ algebra.   We provided the spin decomposition of $\SU(1,1)$-quasi boson spin particle and showed that it is just the Iwasawa decomposition when the fine structure constant $\mbox{\CYRZH}=1$. We constructed the left-invariant Haar measure of the quasi boson and the result is extended to the case of electron in a magnetic field. 
Finally, we demonstrated that the spherical Fourier transformation of the type delta of a $\SU(1,1)$-quasi boson is a Kangni-type transform when the Planck constant, $\hbar=1$.

\bigskip


\section*{Declarations}

\subsection*{Conflict of interest}
The authors declare that they have no competing interests.

\bibliographystyle{spmpsci}
\bibliography{myBibLib} 

\end{document}